\documentclass[10pt,conference]{IEEEtran}
\pdfoutput=1

\usepackage[cmex10]{amsmath}
\usepackage{amssymb,amsfonts}
\usepackage{cite}
\usepackage{color}
\usepackage{graphicx}
\usepackage{subfigure}

\newtheorem{theorem}{Theorem}

\newtheorem{lemma}{Lemma}

\renewcommand{\IEEEQED}{\IEEEQEDopen}

\newcommand{\abs}[1]{\left\vert#1\right\vert}

\begin{document}

\title{Decentralized Two-Hop Opportunistic Relaying With Limited Channel State Information}

\author{\authorblockN{Shengshan Cui and Alexander~M.~Haimovich}
\authorblockA{Department of Electrical and Computer Engineering\\
New Jersey Institute of Technology, Newark, NJ 07102, USA\\
Email: \{shengshan.cui, alexander.m.haimovich\}@njit.edu} \and
\authorblockN{Oren Somekh and H.~Vincent Poor}
\authorblockA{Department of Electrical Engineering\\
Princeton University, Princeton, NJ 08544, USA\\
Email: \{orens, poor\}@princeton.edu}}

\maketitle

\begin{abstract} \boldmath
A network consisting of $n$ source-destination pairs and $m$ relays
with no direct link between source and destination nodes, is
considered. Focusing on the large system limit (large $n$), the
throughput scaling laws of two-hop relaying protocols are studied
for Rayleigh fading channels. It is shown that, under the practical
constraints of single-user encoding-decoding scheme, and partial
channel state information (CSI) at the transmitters (via
integer-value feedback from the receivers), the maximal throughput
scales as $\log n$ even if joint scheduling among relays is allowed.
Furthermore, a novel opportunistic relaying scheme with receiver
CSI, partial transmitter CSI, and decentralized relay scheduling, is
shown to achieve the optimal throughput scaling law of $\log n$.
\end{abstract}

\section{Introduction}
\label{sec:intro}

The ever growing demand for ubiquitous access to high data rate
services necessitates new network architectures, such as \emph{ad
hoc} and relay networks. Over the last decade, a large body of work
analyzing the fundamental system throughput limits of such networks
has been reported. In particular, numerous communication schemes
approaching these limits under various settings have been proposed,
e.g. \cite{GK:00,GHH:06,OLT:07,DH:06,MB:07}.

Notably, Gowaikar \emph{et al.}\ \cite{GHH:06} proposed a new
wireless ad hoc network model, whereby the strengths of the
connections between nodes are drawn independently from a common
distribution, and analyzed the maximum system throughput under
different fading distributions. Such a model is appropriate for
environments with rich scattering but small \emph{physical} size, so
that the connections are governed by random fading instead of
deterministic path loss attenuations (i.e., dense network). When the
random channel strengths follow a Rayleigh fading model, the system
throughput scales as $\Theta(\log n)$.\footnote{Throughout the paper
$\log(\cdot)$ indicates the natural logarithm. For two functions
$f(n)$ and $g(n)$, $f(n)=O(g(n))$ means that $\lim_{n\to \infty}
\abs{f(n)/g(n)} <\infty$, and the notation $f(n)=o(g(n))$ means that
$\lim_{n\to \infty} \abs{f(n)/g(n)} = 0$. We write
$f(n)=\Theta(g(n))$ to denote $f(n)=O(g(n))$ and $g(n)=O(f(n))$.}
This result is achievable through a multihop scheme that requires
central coordination of the routing between nodes.

In this work, we focus on dense networks and two-hop relaying
schemes, in which $n$ source nodes communicate with $n$ destination
nodes via $m$ relay nodes (no direct connection is allowed between
sources and destinations). Dana and Hassibi have proposed an
amplify-and-forward protocol in \cite{DH:06} and shown that a
throughput of $\Theta(n)$  bits/s/Hz is achievable with $m\geq
n^{2}$ relay nodes. It is assumed that each relay node has full
local channel state information (CSI) (backward channels from all
source nodes, and forward channels to all destination nodes), so
that the relays can perform \emph{distributed beamforming}. In
\cite{MB:07}, Morgenshtern and B{\"{o}}lcskei showed a similar
distributed beamforming scheme which demonstrates tradeoffs between
the level of available CSI and the system throughput. In particular,
using a scheme with relays partitioned into groups, where relays
assigned in the same group require knowledge of backward and forward
channels of only one source-destination (S--D) pair, the number of
relays required to support a $\Theta \left( n\right) $ throughput is
$m\geq n^{3}$. Hence, restricting the CSI in such a way increases
the number of required relays from $n^{2}$ to $n^{3}$ to support
throughput of $\Theta(n)$.

While the two-hop schemes reported in \cite{DH:06} and \cite{MB:07}
do not require central coordination among relays (central
coordination is required for the multihop schemes of
\cite{GK:00,GHH:06,OLT:07}), some level of transmitter CSI (channel
amplitude and/or phase) is still required. In a large system,
obtaining this level of CSI, especially at the transmit side, may
not be feasible. 
This consideration leads to the following questions: How does the
throughput scaling change under a practical, partial CSI assumption?
Can the throughput scaling bounds be approached with any specific
schemes?

In the sequel, we give partial answers to the questions above by
restricting ourselves to decode-and-forward protocols. In
Section~\ref{sec:upperBound}, an upper bound on the throughput is
calculated in the large system regime. It is shown that with only
partial CSI at the transmitters, the throughput scaling of any
two-hop scheme is upper-bounded by $\Theta(\log n)$. In
Section~\ref{sec:scheme}, an opportunistic relaying scheme that can
achieve the optimal scaling is proposed. This scheme operates in a
completely decentralized fashion and requires only receiver CSI
knowledge and a low-rate feedback to the transmitters. Finally,
Section~\ref{sec:conclusion} concludes the paper.

\section{Throughput Scaling Upper Bound for Two-Hop Protocols}
\label{sec:upperBound}

In this section, we establish an upper bound on the throughput
scaling of two-hop protocols. We adopt the random connection model
of \cite{GHH:06} and specifically assume a Rayleigh fading model,
i.e., the connections between any source-to-relay (S--R) pair and
between any relay-to-destination (R--D) pair follow independent and
identically distributed (i.i.d.) flat Rayleigh fading. We assume
that in each hop the receivers have perfect CSI knowledge of the
channel realizations, but the transmitters do not have full CSI
knowledge. We assume a single-user encoding-decoding scheme, i.e,
mutual interfering signals are treated as additive noise.
Furthermore, we assume the transmission rate is fixed, i.e., the
transmission rate of each scheduled link is not adaptive to
instantaneous signal-to-interference-plus-noise ratio (SINR).
Accordingly, a transmission is deemed successful only if the SINR is
not below a prescribed threshold.

We have the following throughput upper bound.
\begin{theorem}
\label{thm:upper_bound_p1} Under the aforementioned assumptions, the
throughput of each hop scales {\em at most} as $\log n$.
\end{theorem}

\begin{proof}(Outline)
We begin with the first hop. Since the transmission rate of each
link is a fixed number, finding the throughput upper bound is
equivalent to finding the maximum number of concurrent successful
transmissions. To this end, we consider following genie scheme. For
any channel realization of the network, the genie scheme is assumed
to have the full CSI of the network, and thus is able to schedule in
every time-slot the largest set of concurrent successful S--R pairs.
Specifically, in testing whether $m$ concurrent successful
transmissions are supported or not, the genie scheme will deploy $m$
relays and test whether there exists an $m$-element subset of source
nodes whose transmissions to relays are all successful. In doing so,
the genie scheme will test all $\binom{n}{m}$ ways of choosing $m$
sources for transmission. Moreover, for each combination of $m$
sources, the genie scheme tests $m!$ possible ways of associating
S--R pairs. If the genie scheme can find a combination, among all
$\binom{n}{m}m!$ possible combinations, such that all transmissions
are successful, we claim
that $m$ simultaneous transmissions are achievable. 
By a probabilistic
argument, it is shown in \cite[Th.~2]{CHSP:07} that with probability
approaching $1$, one cannot find a set of $\frac{\log n}{\log 2}+2$
nodes whose simultaneous transmissions to the relays are all
successful. Conversely, with probability approaching $1$, and for
any $\epsilon>0$, there exists a set of $(1-\epsilon )\frac{\log
n}{2\log 2} +2$ nodes whose simultaneous transmissions to relays are
all successful. Since the genie scheme executes an exhaustive search
for maximum number of concurrent successful transmissions, it sets
the upper bound for any decentralized scheme.

Upper bound for the second hop can be derived similarly to the first
hop. There, we seek to find the existence of an $m$-element
destination set such that all $m$ concurrent R--D transmissions are
successful.

The reader is referred to \cite[Th.~2]{CHSP:07} for the complete
proof.
\end{proof}


Theorem~\ref{thm:upper_bound_p1} implies that the throughput scaling
of the two-hop scheme is upper-bounded by the order of $\log n$.
Thus, we are able to answer the first question raised at the outset
of the paper: the lack of full transmitter CSI reduces the
throughput scaling from a power law ($\Theta(n^{1/2})$ \cite{DH:06}
and $\Theta(n^{1/3})$ \cite{MB:07}) to a logarithmic law
$\Theta(\log n)$ if $n$ is interpreted as the total number of nodes
in the network. When implementation is concerned, however, we note
that the genie scheme in the proof requires joint scheduling among
relays and thus is not readily implementable in practice. We are
left with the question as to whether the same throughput scaling is
achievable with practical constraints such as decentralized relay
operation and low rate feedback.

\begin{figure*}[t]
\centerline{ \subfigure[]{\includegraphics{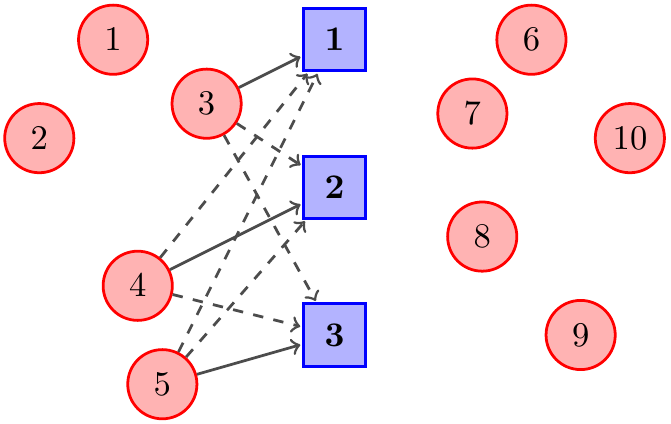}} \hfil
\subfigure[]{\includegraphics{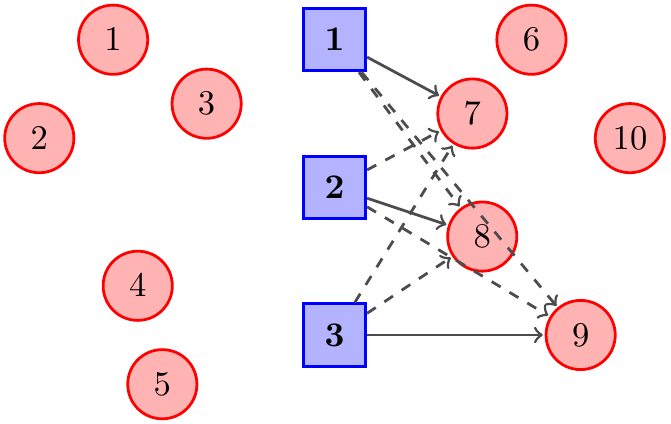}}}
\caption{\label{fig:sys_model}A two-hop network with $n = 5$ S--D
pairs and $m = 3$ relay nodes (denoted by the blue squares). (a) In
the first hop, source nodes $\{3, 4, 5\}$ transmit to the relays.
(b) In the second hop, the relays transmit to the destination nodes.
Solid lines indicate scheduled links, while dashed lines indicate
interferences.}
\end{figure*}

\section{Opportunistic Relaying Scheme}
\label{sec:scheme}

Assuming decentralized relay operation, the relays cannot cancel
mutual interference and have to contend with single-user
encoding-decoding in the two hops. Nevertheless, multiuser diversity
gain, an innate feature of fading channels, is still available and
lends itself to distributed operation. It is shown in the sequel
that, somewhat surprisingly, by exploiting the multiuser diversity,
the throughput scaling of $\log n$ can be achieved with
decentralized
relay operations. 
To enable the scheduling, the scheme requires an index-valued
(integer) CSI from the receivers via low-rate feedback.

\subsection{Scheduling}

As illustrated in Fig.~\ref{fig:sys_model}, the proposed
opportunistic relaying scheme is a two-hop, decode-and-forward-based
communication protocol. In the first hop, a subset of sources is
scheduled for transmission to the relays. Then, the relays decode
and buffer the packets received in the first hop. During the second
hop, the relays forward packets to a subset of destinations (not
necessarily the same set of destinations associated with the sources
set in the first hop). The two phases (hops) are time-interleaved:
Phase $1$ and Phase $2$ take place in even
and odd-indexed time-slots, respectively. 

We assume that the channel gains are dominated by the effects of
small-scale fading. In particular, it is assumed that the wireless
network consists of i.i.d. flat Rayleigh channels. Accordingly, the
channel gain $\gamma_{i,r}$ between the $i$th source node ($1\leq
i\leq n$) and the $r$th relay node ($1\leq r\leq m$), and the
channel gain $\xi_{k,j}$ between the $k$th relay ($1\leq k\leq m$)
and the $j$th destination node ($1\leq j\leq n$), are exponentially
distributed random variables, i.e., $\gamma_{i,r},\
\xi_{k,j}\sim\operatorname{Exp}(1)$.
Quasi-static fading is assumed, in which channels are fixed during
the transmission of each hop, and take on independent values at
different time-slots. We also assume that, in both hops, the
receivers are aware of their backward channel information, and allow
for an integer-value CSI feedback from receivers to transmitters
(relays to sources in Phase $1$, and destinations to relays in Phase
$2$).

\subsubsection{First Hop Scheduling}

The first hop scheduling can be thought of as a natural
generalization of the classic multiuser-diversity-scheme with single
receiver antenna~\cite{KH:95} to multiple, decentralized antennas.
Specifically, all relays operate independently, and each relay
schedules its best source by feeding back the index of the source.
For example, relay $r$ compares the channels
$\gamma_{i^{\prime},r}$, $1\leq i^{\prime}\leq n$, and schedules the
transmission of the strongest source node, say $i=\arg
\max_{i^{\prime }}\gamma_{i',r}$, by feeding back the index $i$. The
overhead of this phase of the protocol is a single integer per relay
node. Suppose the scheduled nodes constitute a set
$\mathcal{K}\subset \{1,\ldots ,n\}$, then since there are $m$
relays, up to $m$ source nodes can be scheduled, i.e.,
$\abs{\mathcal{K}}\leq m$ (a source can be scheduled by multiple
relays). The scheduled source nodes transmit simultaneously at the
same rate of $1$ bit/s/Hz. The communication from source $i$ to
relay $r$ is successful if the corresponding
$\mathsf{SINR}^{\mathrm{P1}}\geq 1$, i.e,
\begin{equation}
\mathsf{SINR}_{i,r}^{\mathrm{P1}}=\frac{\gamma_{i,r}}{1/\rho +
\sum_{\substack{ t\in \mathcal{K} \\ t\neq i}}\gamma_{t,r}}\geq 1,
\label{eq:sinr_ul}
\end{equation}%
where $\rho$ is the average signal-to-noise ratio (SNR) of the S--R
link.

\subsubsection{Second Hop Scheduling}

In the second hop, the transmitters are the $m$ relay nodes, and the
multiuser diversity is achieved by scheduling the destination nodes
via a $\mathsf{SINR}\geq 1$ criterion. In particular, each
destination node $j$, $1\leq j\leq n$, with the assumption of
knowing the forward channel strengths, $\xi _{k,j}$, $1\leq k\leq
m$, computes $m$ SINRs by assuming that relay $k$ is the desired
sender and the other relays are interference:
\begin{equation}
\mathsf{SINR}_{k,j}^{\mathrm{P2}}=\frac{\xi _{k,j}}{1/\rho _{R}+\sum
_{\substack{ 1\leq \ell \leq m \\ \ell \neq k}}\xi _{\ell ,j}},
\label{eq:sinr_dl}
\end{equation}%
where $\rho_{R}$ denotes the average SNR of the R--D link. If the
destination node $j$ captures one good SINR, say,
$\mathsf{SINR}_{k,j}^{\mathrm{P2}}\geq 1$ for some $k$, it instructs
relay $k$ to send data by feeding back the relay index $k$.
Otherwise, the node $j$ does not provide feedback. It follows that
the overhead of the second hop is also at most an index value per
destination node. When scheduled by a feedback message, relay $k$
transmits the data to the destination node at rate $1$ bit/s/Hz. In
case a relay receives multiple feedback messages, it randomly
chooses one destination for transmission.

It is noted that in the steady state operation of the system, the
relays are assumed to buffer the data received from all source
nodes, such that it is available when the opportunity arises to
transmit it to the intended destination nodes over the second hop of
the protocol. This ensures that relays always have packets destined
to the nodes that are scheduled. It should also be noted that, due
to the opportunistic nature of scheduling, the received packets at
the destinations are possibly out of order and therefore each
destination is assumed to have capability of buffering data.

\subsection{Throughput Analysis}

In this subsection, we first derive analytical expressions of the
throughput for each hop assuming the system has a finite number of
nodes. Then, we extract the scaling laws when the system size
increases, i.e., $n\to \infty$, and compare those to the upper
bounds established in the previous section. For the sake of brevity,
we provide here only an outline of the derivation, and the reader is
referred to \cite{CHSP:07} for more details.

\subsubsection{Finite $n$ and $m$}

In the first hop, $m$ relays independently schedule sources. The
number of scheduled sources could be any integer between $1$ and
$m$. Accounting only for the case in which exactly $m$ sources are
scheduled, the average throughput of the first hop can be
lower-bounded as follows,
\begin{equation}
R_{1}\geq m\cdot \Pr [N_{m}]\cdot \Pr [S_{m}],
\label{eq:lower_bound_R_1}
\end{equation}
where $\Pr [N_{m}]$ is the probability of having exactly $m$ sources
scheduled, implying a total transmission rate of $m$ bits/s/Hz. $\Pr
[S_{m}]$ is the probability for a successful S--R transmission.

By symmetry, each source node has a probability of $1/n$ to be the
best node with respect to a relay. Thus, $\Pr [N_{m}]=n(n-1)\cdots
(n-m+1)/n^{m}$. For finite values of $n$ and $m$, exact
characterization of $\Pr [S_{m}]$ is mathematically involved. This
is because the numerator (the maximum of $n$ i.i.d. random
variables) and the denominator (summation of some non-maximum random
variables) are not independent. Fortunately, it is possible to
further lower-bound $\Pr [S_{m}]$ as follows,
\begin{align}
\Pr [S_{m}]& =\Pr [\mathsf{SINR}^{\mathrm{P1}}\geq 1] =\Pr \biggl[\frac{X}{1/\rho +Y}\geq 1\biggr]  \notag \\
           & =\Pr [X\geq s]\cdot \Pr \biggl[\frac{X}{1/\rho +Y}\geq 1\Bigl|X\geq s \biggr]  \notag \\
                & \quad +\Pr [X\leq s]\cdot \Pr \biggl[\frac{X}{1/\rho +Y}\geq 1\Bigl|X\leq s\biggr]  \notag \\
           & \geq \Pr [X\geq s]\cdot \Pr \biggl[\frac{X}{1/\rho +Y}\geq 1\Bigl|X\geq s\biggr]  \notag \\
           & \geq \Pr [X\geq s]\cdot \Pr \biggl[\frac{s}{1/\rho +Y}\geq 1\biggr]  \notag \\
           & =\bigl(1-F_{X}(s)\bigr)F_{Y}(s-1/\rho ),  \label{eq:lower_bound_PrSm}
\end{align}
where $X$ represents the maximum of $n$ i.i.d. exponential random
variables, whose cumulative distribution function (CDF) can be
written explicitly as $F_{X}(x)=(1-e^{-x})^{n}$. The term
$F_Y(\cdot)$ denotes the CDF of the aggregate interference, which is
shown in \cite{CHSP:07} to be well approximated to a chi-square
random variable with $2(m-1)$ degrees-of-freedom with CDF
$F_{Y}(y)=1-e^{-y}\sum_{k=0}^{m-2}\frac{1}{k!}\,y^{k}$, when $n$ is
sufficiently large, e.g., $n>40$. Note that the lower bound
\eqref{eq:lower_bound_PrSm} suggests a suboptimal scheduling scheme
according to which, each relay schedules the transmission of the
``strongest'' source {\em only} if the source's power gain exceeds a
prescribed threshold $s$. The probability of such event is given by
$1-F_{X}(s)$, and $F_{Y}(s-1/\rho)$ is a lower bound on the
probability of a successful communication with the relay at a rate
of $1$ bit/s/Hz.

Substituting the lower bound of $\Pr[S_m]$ into
\eqref{eq:lower_bound_R_1}, we get a lower bound on $R_1$, as
expressed in the following lemma.

\begin{lemma}
\label{lem:fix_m_large_n_ul} For any $\rho$, $n> m$ and $s>0$, the
achievable throughput of the opportunistic relay scheme of the first
hop is lower-bounded by
\begin{equation}
R_{1}\geq m\tfrac{n(n-1)\cdots (n-m+1)}{n^{m}}\bigl(1-(1-e^{-s})^{n}\bigr) %
F_{Y}\Bigl(s-\tfrac{1}{\rho }\Bigr).  \label{eq:llower_bound_R_1}
\end{equation}
\end{lemma}
\smallskip

Turning to the second hop and recalling that its scheduling is based
on SINR instead of SNR, all transmissions are successful by
definition. Thus, the throughput of the second hop depends on how
many relays receive feedback and therefore transmit data packets to
the destinations. Furthermore, a relay is scheduled when at least
one destination measures its channel with SINR greater than or equal
to one. Therefore, the average throughput can be characterized in a
closed form expression, as formulated in the lemma.

\begin{lemma}
\label{lem:fix_m_large_n_dl} For any $\rho_R$, $n$ and $m$, the
achievable throughput of the opportunistic relay scheme in the
second hop is given by
\begin{equation}
R_{2}=m\left( 1-\left( 1-\frac{e^{-1/\rho _{R}}}{2^{m-1}}\right) ^{n}\right).\label{eq:R2}
\end{equation}
\end{lemma}
\smallskip

\subsubsection{Large $n$ and Finite $m$}
\label{sec:oppor_fixed_m}

With the closed-form expressions of \eqref{eq:llower_bound_R_1} and
\eqref{eq:R2} at hand, we proceed to the regime of large $n$, but
fixed $m$. The discussion of this regime is of practical importance
in that as communication devices become pervasive, the number of
infrastructure nodes (here the relays) is not likely to keep pace.

As mentioned above, the parameter $s$ in \eqref{eq:llower_bound_R_1}
can be interpreted as a scheduling threshold. Note that in a system
with $n$ sources and Rayleigh fading channels, the maximum channel
gain seen by each relay is of the order of $\log n$ \cite{VTL:02},
we empirically set $s=\log n-\log\log n$ in
\eqref{eq:llower_bound_R_1}. Then, it is easy to show that $R_1\to
m$ with $n\to \infty$. Similarly, letting $n\to \infty$ in
\eqref{eq:R2}, results in $R_2\to m$. Now, since the average
throughput of the two-hop scheme is $R=\frac{1}{2}\min\{R_1,R_2\}$,
we conclude that $R\to \frac{m}{2}$ for $n\to \infty$. The results
for large $n$ and finite $m$ are summarized in the following
theorem.

\begin{theorem}
\label{thm:thput_fix_m_large_n} For fixed $m$, the two-hop
opportunistic relaying scheme achieves a system throughput of $m/2$
bits/s/Hz as $n\rightarrow \infty $.
\end{theorem}

In the opportunistic scheme, we make practical assumptions of
decentralized relays and partial CSI. Thus, it is instructive to
compare the throughput of the opportunistic scheme with that of an
unconstrained scheme. In fact, it is straightforward to show that,
the information-theoretical sum-rate for {\em any} two-hop scheme is
upper-bounded by $\frac{m}{2}\log \log n$ \cite[Lemma~3]{CHSP:07},
even if relay cooperation and full CSI at the relays are assumed.
This upper bound (with cooperation and full CSI) can be interpreted
as a multiple antenna system, which is well-known to be able to
support $m$ parallel channels. Moreover, {\em each} of the parallel
channels enjoys multiuser diversity gain of $\log n$ that translates
into a throughput of $\log\log n$. In contrast, the opportunistic
scheme, with simplified network operation (decentralized operation
and partial CSI assumption), has no such freedom to support $m$
parallel channels with rate $\log\log n$. However, it succeeds in
preserving the pre-log factor of the upper bound. Intuitively, the
inherent multiuser diversity gain, which is of the order of $\log
n$, is applied to compensate for the mutual interference stemming
from concurrent transmissions and to make the scheduled links
reliable.

\subsubsection{Large $n$ and $m$}
\label{sec:oppor_large_m}

Theorem~\ref{thm:thput_fix_m_large_n} shows that when the number of
S-D pairs $n$ is large and the number of relay nodes $m$ is fixed,
the average system throughput scales linearly with $m$. This implies
that one can increase the number of relays to increase system
throughput. However, both \eqref{eq:llower_bound_R_1} and
\eqref{eq:R2} present a tradeoff of throughput in $m$: by making $m$
large, one increases the number of transmissions, but as a
consequence the reliability of each link degrades. Therefore, there
exists an optimal value of $m$ such that the throughput scaling is
maximized. Finding the optimal order of $m$ is equivalent to finding
the throughput scaling of the proposed opportunistic relaying
scheme. Specifically, we are interested in finding whether the
proposed scheme can achieve the throughput scaling upper bound of
$\Theta(\log n)$ established in Section~\ref{sec:upperBound}.

To prove that the average throughput of the first hop indeed scales
as $\Theta(\log n)$, it is sufficient to show that the lower bound
\eqref{eq:llower_bound_R_1} achieves scaling of order $\log n$. To
this end, consider the case of $m=\log n$ and $s=\log n-\log \log
n$. With $n\rightarrow \infty$, it follows that $\tfrac{n(n-1)\cdots
(n-m+1)}{n^{m}} \rightarrow 1$ and
$\bigl(1-(1-e^{-s})^{n}\bigr)\rightarrow 1$. Furthermore, for
$m=\log n$, the interference term $Y$ can be approximated by a
Gaussian random variable with mean and variance both equal to $\log
n$. Due to the symmetry of the Gaussian distribution, we have
$F_{Y}(\log n-\log \log n-1/\rho )\approx F_{Y}(\log
n)=\frac{1}{2}$. This result implies that if we deploy $m=\log n$
relays, with high probability, $\log n$ sources will be scheduled
for transmission, and half of them will be, on average, successful.
This yields an average throughput of $\frac{1}{2}\log n$ for the
first hop.

Examining the asymptotic behavior of \eqref{eq:R2} with respect to
$m$ and $n$, it is straightforward to show that the maximum
throughput scaling of the second hop also scales as $\Theta(\log n)$
\cite[Th.~3]{CHSP:07}:

\begin{theorem}
\label{thm:Ptwo_large_m} For the second hop of the two-hop
opportunistic relaying scheme, if the number of relays $m=\frac{\log
n-\log \log n-1/\rho_R}{\log 2}+1$, then $R_{2}=\Theta \left(
m\right) =\Theta \left( \log n\right) $. Conversely, if
$m=\frac{\log n+\log \log n-1/\rho_R}{\log 2}+1$, then $R_{2}=o(m)$.
\end{theorem}

By considering two hops as a whole, we get the following:

\begin{theorem}
\label{thm:max_thput_our_scheme} Under the setup of
Section~\ref{sec:scheme}, the proposed two-hop opportunistic
relaying scheme yields a maximum achievable throughput of $\Theta
\left( \log n\right) $.
\end{theorem}

Interestingly, we see that the proposed opportunistic relaying
scheme, which assumes decentralized relay operations and practical
CSI assumption, incurs no loss in achieving the optimal throughput
scaling upper bound. This gives an affirmative answer to the second
question posed at the outset of the paper.

The achievability of $\Theta(\log n)$ is also substantiated by Monte
Carlo simulations. In the simulations, the average SNR of each hop
is assumed to be $10$ dB and the simulation curve was obtained by
averaging throughput over $2000$ channel realizations. In
Fig.~\ref{fig:sim}, the average system throughput of the two-hop
opportunistic relaying scheme is shown as a function of the number
of S--D pairs $n$. (Note that the throughput depends on both $n$ and
$m$. For each value of $n$, optimal throughput (by maximizing over
$m$) is plotted.) We observe that the throughput exhibits the $\log
n$ trend, as predicted by Theorem~\ref{thm:max_thput_our_scheme}. It
is also found in simulation that the system throughput is always
limited by Phase $1$, i.e., $R=\frac{1}{2}\min \{R_{1},R_{2}\}=
\frac{1}{2}R_{1}$. Thus, we also plot $1/2$ of the upper bound and
lower bound of $R_{1}$ for reference. Recall that the average
throughput of $R_1$ is upper-bounded by the genie bound $\frac{\log
n}{2\log 2}+2$ (cf.~Theorem~\ref{thm:upper_bound_p1}) and
lower-bounded by $\frac{1}{2}\log n$.

\begin{figure}[htb!]
\centering
\includegraphics[width=8.5cm]{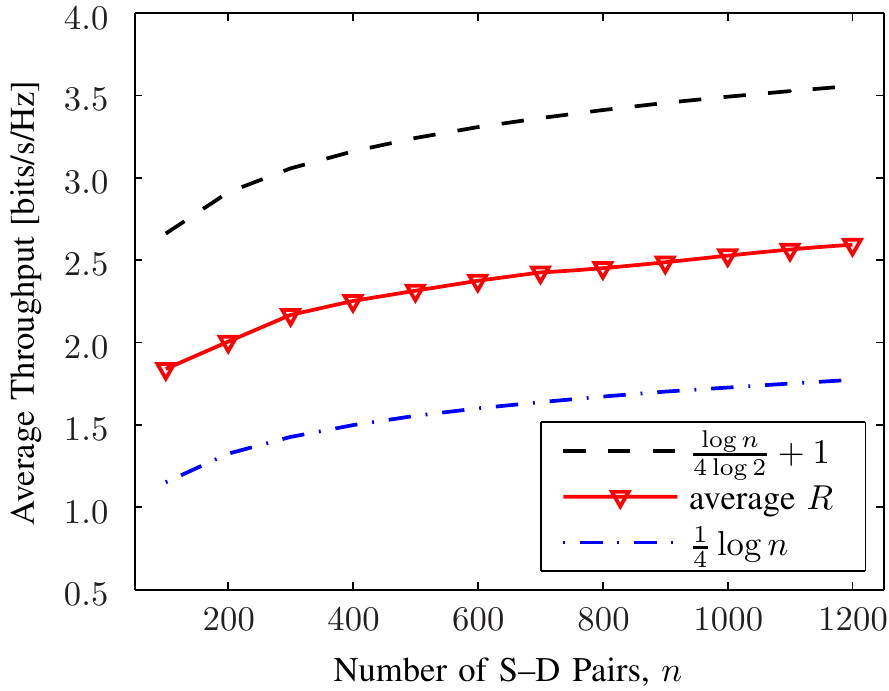}
\caption{Simulated average system throughput $R$ and analytical
upper and lower bounds. \label{fig:sim}}
\end{figure}

\subsection{Feedback Overhead}

According to the opportunistic relaying scheme, a feedback mechanism
is needed to schedule the {\em good} nodes enjoying multiuser
diversity. By direct computation, it can be shown that, in the
limiting operation regime of $m=\Theta(\log n)$, the feedback
overhead {\em per fading block} is $\Theta((\log n)^2)$ for the
first hop and $\Theta(\log n \log\log n)$ for the second hop. The
overhead of feedback is negligible when the block length is large.

\section{Conclusion}
\label{sec:conclusion}

In this paper, we have considered a network having $n$ S--D pairs
and $m$ relay nodes, operating in the presence of Rayleigh fading.
The emphasis is on characterizing the throughput scaling under the
assumption of practical CSI requirement. It has been shown that the
lack of full CSI at the relays reduces the throughput scaling
drastically from a power law (e.g., $\Theta(n^{1/2})$ \cite{DH:06})
to a logarithmic law $\Theta(\log n)$ in the total number of nodes
$n$ in the network.

Furthermore, an opportunistic relaying scheme that operates in a
completely decentralized fashion and assumes only CSI at receivers
and partial CSI at the transmitters, has been proposed and shown to
achieve a throughput scaling of $\Theta(\log n)$. Thus, the lack of
joint scheduling among relays causes no loss of optimality as far as
throughput scaling is concerned.

An interesting subject for further research is the performance
analysis of opportunistic relaying schemes employed in more general
system models. In particular, models that include both small-scale
fading and geographical attenuation (e.g. the model presented in
\cite{GH:06}) are of interest.

\section*{Acknowledgment}

This research was supported in part by the U.S. National Science
Foundation under Grants CNS-06-26611, CNS-06-25637 and
ANI-03-38807, and by a Marie Curie Outgoing International
Fellowship within the 6th European Community Framework Programme.

\bibliographystyle{IEEEtran}
\bibliography{ISIT08_final}

\begin{thebibliography}{1}
\providecommand{\url}[1]{#1}
\csname url@samestyle\endcsname
\providecommand{\newblock}{\relax}
\providecommand{\bibinfo}[2]{#2}
\providecommand{\BIBentrySTDinterwordspacing}{\spaceskip=0pt\relax}
\providecommand{\BIBentryALTinterwordstretchfactor}{4}
\providecommand{\BIBentryALTinterwordspacing}{\spaceskip=\fontdimen2\font plus
\BIBentryALTinterwordstretchfactor\fontdimen3\font minus
  \fontdimen4\font\relax}
\providecommand{\BIBforeignlanguage}[2]{{%
\expandafter\ifx\csname l@#1\endcsname\relax
\typeout{** WARNING: IEEEtran.bst: No hyphenation pattern has been}%
\typeout{** loaded for the language `#1'. Using the pattern for}%
\typeout{** the default language instead.}%
\else
\language=\csname l@#1\endcsname
\fi
#2}}
\providecommand{\BIBdecl}{\relax}
\BIBdecl

\bibitem{GK:00}
P.~Gupta and P.~R. Kumar, ``The capacity of wireless networks,'' \emph{{IEEE}
  Trans. Inf. Theory}, vol.~46, no.~2, pp. 388--404, Mar. 2000.

\bibitem{GHH:06}
R.~Gowaikar, B.~M. Hochwald, and B.~Hassibi, ``Communication over a wireless
  network with random connections,'' \emph{{IEEE} Trans. Inf. Theory}, vol.~52,
  no.~7, pp. 2857--2871, Jul. 2006.

\bibitem{OLT:07}
A.~{\"O}zg{\"u}r, O.~L\'{e}v\^{e}que, and D.~N.~C. Tse, ``Hierarchical
  cooperation achieves optimal capacity scaling in ad hoc networks,''
  \emph{{IEEE} Trans. Inf. Theory}, vol.~53, no.~10, pp. 3549--3572, Oct. 2007.

\bibitem{DH:06}
A.~F. Dana and B.~Hassibi, ``On the power efficiency of sensory and ad hoc
  wireless networks,'' \emph{{IEEE} Trans. Inf. Theory}, vol.~52, no.~7, pp.
  2890--2914, Jul. 2006.

\bibitem{MB:07}
V.~I. Morgenshtern and H.~B{\"o}lcskei, ``Crystallization in large wireless
  networks,'' \emph{{IEEE} Trans. Inf. Theory}, vol.~53, no.~10, pp.
  3319--3349, Oct. 2007.

\bibitem{CHSP:07}
\BIBentryALTinterwordspacing
S.~Cui, A.~M. Haimovich, O.~Somekh, and H.~V. Poor, ``Opportunistic relaying in
  wireless networks,'' \emph{{IEEE} Trans. Inf. Theory}, Dec. 2007, submitted
  for publication. [Online]. Available: \url{http://arxiv.org/pdf/0712.1169}
\BIBentrySTDinterwordspacing

\bibitem{KH:95}
R.~Knopp and P.~Humblet, ``Information capacity and power control in single
  cell multiuser communications,'' in \emph{Proc. {IEEE} Int. Conf.
  Communications}, vol.~1, Seattle, WA, Jun. 1995, pp. 331--335.

\bibitem{VTL:02}
P.~Viswanath, D.~N.~C. Tse, and R.~Laroia, ``Opportunistic beamforming using
  dumb antennas,'' \emph{{IEEE} Trans. Inf. Theory}, vol.~48, no.~6, pp.
  1277--1294, Jun. 2002.

\bibitem{GH:06}
R.~Gowaikar and B.~Hassibi, ``On the achievable throughput in two-scale
  wireless networks,'' in \emph{Proc. {IEEE} Int. Symp. Information Theory},
  Seattle, WA, Jul. 2006.

\end{thebibliography}

\end{document}